\documentclass[aps,10pt,twocolumn,showpacs,pra,citeautoscript,amsmath,amssymb,floatfix,nofootinbib,superscriptaddress,longbibliography]{revtex4-1}

\usepackage{amsmath}
\usepackage{amssymb}
\usepackage{bbm}
\usepackage{epsfig}
\usepackage{color,epstopdf}
\usepackage{amsthm, thm-restate}

\usepackage{natbib}
\usepackage{hyperref}
\usepackage{cleveref}
\crefname{figure}{figure}{figures}

\usepackage{soul}
\usepackage[caption=false]{subfig}
\usepackage{algpseudocode}
\usepackage{enumerate}

\newcounter{lemmaN}

\newtheorem{lemma}[lemmaN]{Lemma}

\newtheorem*{theorem*}{Main Theorem}

\newcounter{colN}

\newtheorem{corollary}[colN]{Corollary}

\newtheorem{definition}{Definition}

\newcommand{\ket}[1]{|#1\rangle}

\newcommand{\ketbra}[2]{|#1\rangle\langle #2|}
\DeclareMathOperator{\tr}{tr}
\newcommand{\id}{\mathbbm{1}}

\newcommand{\app}[2]{\left( #1, #2 \right)}
\newcommand{\inn}[2]{\langle #1, #2 \rangle}

\newcommand{\reals}{\mathbb{R}}
\newcommand{\comp}{\mathbb{C}}

\newcommand{\trans}[0]{^{\rm T}}

\newcommand{\conj}[0]{^{*}}

\newcommand{\Csa}[1]{\mathbf{H}_{#1}(\mathbb{C})}
\newcommand{\Csap}[1]{\mathbf{H}_{#1}^{+}(\mathbb{C})}

\newcommand{\jord}[0]{\bullet}

\newcommand{\sjord}[0]{\bullet}

\newcommand{\inlineheading}[1]{\textbf{{#1}}}

\newcommand{\IQOQI}{Institute for Quantum Optics and Quantum Information,\\ Austrian Academy of Sciences, Boltzmanngasse 3, A-1090 Vienna, Austria}
\newcommand{\Peri}{Perimeter Institute for Theoretical Physics, 31 Caroline Street North, Waterloo, ON N2L 2Y5, Canada}

\begin{document}
\title{Characterization of the probabilistic models that can be embedded in quantum theory}

\author{Andrew J.\ P.\ Garner}
\affiliation{\IQOQI{}}
\author{Markus P.\ M\"uller}
\affiliation{\IQOQI{}}
\affiliation{\Peri{}}

\date{April 13, 2020}

\begin{abstract}
Quantum bits can be isolated to perform useful information-theoretic tasks,
 even though physical systems are fundamentally described by very high-dimensional operator algebras.
This is because qubits can be consistently embedded into higher-dimensional Hilbert spaces. 
A similar embedding of classical probability distributions into quantum theory enables the emergence of classical physics via decoherence. 
Here, we ask which other probabilistic models can similarly be embedded into finite-dimensional quantum theory.
We show that the embeddable models are exactly those that correspond to the Euclidean special Jordan algebras: quantum theory over the reals, the complex numbers, or the quaternions, and ``spin factors'' (qubits with more than three degrees of freedom), and direct sums thereof. 
Among those, only classical and standard quantum theory with superselection rules can arise from a physical decoherence map.
Our results have significant consequences for some experimental tests of quantum theory, by clarifying how they could (or could not) falsify it. 
Furthermore, they imply that all unrestricted non-classical models must be contextual.
\end{abstract}

\maketitle

\noindent
\inlineheading{Introduction.} 
What can be {\em embedded} into a physical quantum system?
Answering this is crucial for understanding the full variety of information processing possible in quantum physics.
Qubits -- two level quantum systems -- rarely exist in nature, but are often isolated within larger quantum systems (e.g.,\ by choosing just two energy levels of a cold atom's spectrum). 
Meanwhile, quantum error correcting codes embed lower-dimensional systems into higher-dimensional ones to improve resilience to errors; 
 a similar embedding of classical probability theory into quantum theory is a prerequisite for the emergence of classical physics via decoherence processes.
What other logical state spaces could be mapped onto that cold atom? 
What other systems can we find quantum encodings for?

In this letter, we completely characterize the probabilistic models that can be embedded into finite-dimensional quantum theory. 
It turns out that these correspond exactly to a family of structures that are well-known in the foundations of quantum mechanics: the Euclidean special Jordan algebras~\cite{JordanvNW34,AlfsenS03,McCrimmon04,Baez12,Wilce09,Wilce12,BarnumWilce14,BarnumH19}. 
These models arise from the Hilbert spaces over the real, complex and quaternionic~\cite{Graydon11} fields, from state spaces that are $d$-dimensional ``Bloch balls''~\cite{PawlowskiW12,PaterekDB10,KrummM19,GarnerMD17}, and from their direct sums (including classical probability theory and quantum theory with superselection rules~\cite{BartlettRS07}).
  
When we consider the projections onto these embeddings, we find that only complex quantum theory with superselection rules can arise via a physically-realizable completely positive map,
 suggesting why we only typically find quantum and classical behaviours in nature. Furthermore, we consider the models that can also be embedded into {\em classical theory}. We determine that these are only those that are classical themselves, and discuss the implications of this for the study of contextuality in operational theories~\cite{Spekkens05,SchmidSWKW19,Shahandeh19}.

\newpage
\inlineheading{Framework: effects and states.} 
Probabilistic models describe \emph{operational theories} -- theories in which laboratory operations such as preparation and measurement procedures are the primitive elements~\cite{Spekkens05}.
This is mathematically modelled by an ordered vector space~\cite{AliprantisT07} $\left(A,A_+,u_A\right)$, where $A$ is a finite-dimensional vector space, $A_+\subset A$ is the closed and generating cone\footnote{A cone $K$ is a non-empty convex subset of a vector space, satisfying $K+K\subseteq K$, $\alpha K \subseteq K$ for all $\alpha\geq 0$, and $K\cap (-K)= \{0\}$~\cite{AliprantisT07}. 
It is generating if it spans the full vector space (otherwise we would restrict our attention to a subspace). 
For a physical motivation of topological closure, see~\citet{MasanesMueller2011}.
} of {\em positive} elements and $u_A\in A_+$ is the order unit.
The elements $a\in A_+$ are the (unnormalized) \emph{effects}, and they correspond to the possible outcomes that can appear in any measurement of the operational theory. 
In particular, the unit effect $u_A$ corresponds to the affirmative answer to the question ``is the system there?''.
Thus, an ($n$-outcome) measurement is a collection of effects $\{a_i\}_{i=1\ldots n}$ satisfying $\sum_{i=1}^n a_i = u_A$.

A {\em state} is a functional $\omega$ that represents the possible statistics arising from a class of preparation procedures.
That is, a state assigns probabilities to every effect $a\in A_+$ by $\app{\omega}{a} := \omega\!\left(a\right)$. 
As probabilistic mixtures of preparations must yield the correctly-weighted outcome probability~\cite{Barrett07}, $\omega$ is a linear functional, i.e.\ an element of the dual vector space $A\conj$. 
Then, $\mathcal{A}$'s set of possible (unnormalized) states is also a cone, $A\conj_+ := \{\omega \in A\conj \;|\; \left(\omega, e\right) \geq 0 \forall \,e\,\in A_+ \}$. 
A state $\omega\in A_+\conj$ is normalized if $\app{\omega}{u_A} = 1$, and the convex set of all normalized states is denoted $\Omega_A$. This formalism adopts the \emph{no-restriction hypothesis}~\cite{JanottaLal13}: all those objects which give non-negative probabilities on all effects are valid states, and vice versa. 

In this language, an $n$-level \emph{quantum model} $\mathcal{Q}_n$ 
 has the effect space $(B,B_+,u_B)$, where $B=\Csa{n}$ is the vector space of complex self-adjoint (i.e.\ Hermitian) $n\times n$ matrices, $B_+ =\Csap{n}$ is the cone of positive--semidefinite Hermitian matrices (i.e.\ non-negative multiples of POVM elements~\cite{NielsenC00}), and $u_B=\id_n$ is the $n\times n$ identity matrix. 
We identify $B$ and $B\conj$ via the Hilbert--Schmidt inner product $\inn{x}{y} := {\rm tr}(x y)$, and as a result, $B\conj_+ = B_+ =\Csap{n}$.
Quantum states are thus also represented by positive semidefinite Hermitian {\em density matrices}, and $\Omega_B$ consists of normalized states such that $\tr\rho = 1$.
Meanwhile, an $n$-level classical model $\mathcal{C}_n$ has an effect space $B' :=  \reals^n$, 
 with a positive simplex cone $B'_+:=S_n$ of vectors with no negative elements, and the unit effect $u_{B'} = \left(1, \ldots, 1\right)\trans$. 
Via the usual `dot' inner product, the classical state space is also $\reals^n$, and the normalized states are the probability vectors $(p_1,\ldots,p_n)$ with $p_i\geq 0$ and $\sum_i p_i=1$. 
Finally, quantum models with superselection rules~\cite{BartlettRS07} correspond to effect spaces $\bigoplus_i \Csa{n_i}$, which can be thought of as block matrices in some basis. These contain classical models as special cases, since $\reals^n\simeq \bigoplus_{i=1}^n \Csa{1}$.

\inlineheading{Framework: Jordan algebras}.
A Jordan algebra $\left(\mathcal{J},\jord\right)$ consists of some set $\mathcal{J}$ and a product $\jord$ that is commutative ($x\jord y = y \jord x$ for all $x,y\in\mathcal{J}$) and satisfies the Jordan identity ($\left(x \jord y\right)\jord\left(x\jord x\right) = x\jord(y\jord(x\jord x))$ for all $x,y\in\mathcal{J}$). For reasons that will become clear, we shall restrict our discussion to finite-dimensional $\mathcal{J}$. Such $\mathcal{J}$ is Euclidean if there exists an inner product with the  property that $\inn{x}{z\sjord y} = \inn{z\sjord x}{y}$ for all $x,y,z \in \mathcal{J}$~\cite{FarautK94}. The cone of positive elements in a Euclidean Jordan algebra, $\mathcal{J}_+$, is defined as the set of all squares $\mathcal{J}_+ :=  \{x^2\,\,|\,\,x\in\mathcal{J}\}$, where $x^2:=x\sjord x$. The \emph{Jordan unit} is the unique element $u_{\mathcal{J}}\in\mathcal{J}$ with the property that $u_{\mathcal{J}} \jord x =x$ for all $x\in\mathcal{J}$. We can view $\mathcal{J}$ as a probabilistic model, with effect space $(\mathcal{J},\mathcal{J}_+,u_{\mathcal{J}})$. A Jordan algebra is \emph{special} if it is isomorphic to a Jordan subalgebra of an associative algebra~\cite{McCrimmon04}.

For example, the set of complex self-adjoint matrices $\Csa{n}$ equipped with the product $x\sjord y := \frac{1}{2}\left(xy + yx\right)$ for $x,y\in\Csa{n}$ (where $xy$ is standard matrix multiplication) is a Euclidean special Jordan algebra, where the Euclidean inner product is the usual Hilbert--Schmidt product $\inn{x}{y} =\tr\left(xy\right)$.

\inlineheading{Simulations and embeddings}.
Suppose we simulate the statistical behavior of model $\mathcal{A}$ using another model $\mathcal{B}$  -- that is, the preparation of any state $\omega_A$ of $\mathcal{A}$ can be simulated by some corresponding preparation $\omega'_B$ on $\mathcal{B}$ (and likewise every measurement with outcome effect $a$ can be simulated by the corresponding effect $b'$ of $\mathcal{B}$).
Then, there will be two maps: a map $\psi:A_+^*\to B_+^*$ that assigns to every state $\omega_A$ the corresponding state $\omega'_B=\psi(\omega_A)$; and a map $\varphi:A_+\to B_+$ that assigns to every effect $a$ the corresponding effect $b'=\varphi(a)$,
 such that these maps preserve all outcome probabilities [ $b'(\omega'_B)=a(\omega_A)$ ].

This formulation encodes an assumption of {\em procedural independence}: while there may be different operational procedures to prepare a state $\omega_A$ of $\mathcal{A}$, we assume that all of them will be simulated with the \emph{same} state $\omega'_B$ of $\mathcal{B}$. 
When taken with the convexity of state (and effect) spaces, this has a particular consequence.
Consider, for example, preparing $\rho_A:=\frac{1}{2} \omega_A^{(1)}+\frac{1}{2} \omega_A^{(2)}$ by tossing an ideal coin and preparing either $\omega_A^{(1)}$ or $\omega_A^{(2)}$, depending on the outcome. 
This means that we can prepare ${\rho}'_B$ in $\mathcal{B}$ that reproduces the statistics of ${\rho}_A$ by tossing a coin and preparing either ${\omega'}_B^{(1)}$ or ${\omega'}_B^{(2)}$, depending on the outcome, where ${\omega'}_B^{(i)}=\psi(\omega_A^{(i)})$. 
Procedural independence then implies that ${\rho}'_B=\psi(\rho_A)$ -- that is, $\psi$ must be a convex-linear map. 
We can then extend $\psi$ linearly to the full space $A^*$, and similar argumentation applies to the effect map $\varphi$. We hence define such procedurally independent simulations that also preserve normalization as an {\em embedding}:

\begin{definition}[Embedding]
\label{def:EmbeddingMap}
For models $\mathcal{A}$ and $\mathcal{B}$ with respective effect spaces $\left(A,A_+,u_A\right)$ and $\left(B,B_+,u_B\right)$ and respective dual state spaces $\left(A\conj, A\conj_+,\Omega_A\right)$ and $\left(B\conj, B\conj_+,\Omega_B\right)$, a pair of linear maps $\varphi: A \to B$ and $\psi: A\conj \to B\conj$ is said to {\bf embed} $\mathcal{A}$ into $\mathcal{B}$ if:-
\begin{enumerate}[(i)]
\item $\varphi$ and $\psi$ are {\em positive} ($\varphi(A_+)\subseteq B_+$ and $\psi(A\conj_+)\subseteq B\conj_+$) and $\varphi$ is {\em unital} ($\varphi(u_A)=u_B$),
 \item $\varphi$ and $\psi$ {\em preserve outcome probabilities}; i.e.\ for all $e\in A_+$, $\omega \in A_+\conj$, $(\omega,e)=(\psi(\omega),\varphi(e))$.
\end{enumerate}
\end{definition}

This definition has a few immediate consequences. 
First, by linearity (ii) will hold also for all $e\in A$ and $\omega\in A^*$. Second, $(\psi(\omega),u_B)=(\psi(\omega),\varphi(u_A))=(\omega,u_A)$, and hence $\psi(\Omega_A)\subseteq \Omega_B$.
 Furthermore,
\begin{lemma}
\label{lem:EmbeddingFacts}
For an embedding $\mathcal{A}\to\mathcal{B}$ with maps $\varphi:A\to B$ and $\psi:A\conj \to B\conj$ (as per \cref{def:EmbeddingMap}):
\begin{enumerate}[(i)]
\item $\psi\conj\varphi=\mathbf{1}_A$ and $\psi\conj\!\left(B_+\right) \subseteq A_+$ -- i.e.\ the dual of $\psi$ is a positive left-inverse of $\varphi$, and thus $\dim A \leq \dim B$. Likewise, $\varphi\conj\psi = \mathbf{1}_{A\conj}$ and $\varphi\conj (B\conj_+) \subseteq A\conj_+$.
\item The map $P:=\varphi\psi\conj: B \to B$ is a positive unital projection onto the image of effects $\varphi\!\left(A\right)$ ($P^2 = P$, $P B_+ \subseteq B_+$, $P u_B = u_B$, and $P(B) =\varphi(A)$).
Similarly, $P\conj = \psi\varphi\conj: B\conj \to B\conj$ is a positive projection onto the image of states $\psi\!\left(A\conj\right)$ ($P^{*2} = P\conj$, $P\conj B\conj_+ \subseteq B\conj_+$, and $P^*(B\conj)=\psi(A^*)$).
\item $\varphi(A_+)=\varphi(A)\cap B_+=P(B_+)$.
\end{enumerate}
\begin{proof}
{\bf (i):} 
$(\omega,e)=(\psi(\omega),\varphi(e))=(\omega,\psi^* \varphi(e))$ for all $e\in A, \omega\in A\conj$, and therefore $\psi\conj\varphi=\mathbf{1}_A$.
The positivity of $\psi\conj$ follows from that of $\psi$:
Let $b\in B_+$, such that $(x,b) \geq 0$ for all $x\in B\conj_+$.
Then, $(\psi(a),b) \geq 0$ for all $a\in A\conj_+$
 and hence $(a,\psi^*(b))\geq 0$ for all $a\in A\conj_+$, implying that $\psi\conj(b)\in A_+$, and hence $\psi^*$ is positive.
Similar holds for the dual.
{\bf (ii):} First, $P^2 = \varphi\psi\conj\varphi\psi\conj = \varphi\psi\conj = P$. 
As both $\varphi$ and $\psi\conj$ are positive, so is $P$.
As $\psi\conj\varphi \equiv \mathbf{1}_A$, then $P \varphi\!\left(a\right) = \varphi \psi\conj  \varphi \!\left(a\right) = \varphi\!\left(a\right)$ for all $a\in A$. 
This shows that $\varphi(A)\subseteq {\rm im}\, P$. Conversely, if $b\in {\rm im}\, P$ define $a:=\psi^*(b)$, then $\varphi(a)=\varphi\psi^*(b)=Pb=b$, i.e.\ $b\in\varphi(A)$.
To show unitality, 
 apply $\varphi$ to both sides of $\psi\conj\varphi\!\left(u_A\right) = u_A$ to yield $\varphi \psi\conj \varphi\!\left(u_A\right) = \varphi\!\left(u_A\right)$,
  then use $\varphi\!\left(u_A\right) = u_B$ to conclude $\varphi \psi\conj u_B = u_B$.
Similar reasoning establishes $P\conj$ as a positive projector onto $\psi\!\left(A\conj\right)$.
{\bf (iii):} $\varphi(A_+)\subseteq \varphi(A)\cap B_+$ is trivial, and if $b\in\varphi(A)\cap B_+=P(B)\cap B_+$ then $b=Pb\in P(B_+)$. For the converse inclusions, we have $P(B_+)\subseteq B_+$ due to positivity of $P$ and $P(B_+)\subseteq P(B)=\varphi(A)$. To see that $\varphi(A)\cap B_+\subseteq \varphi(A_+)$, let $a\in A$ and $\varphi(a)\in B_+$, then for all $\omega\in A_+\conj$, we have $(\omega,a)=(\psi(\omega),\varphi(a))\geq 0$, hence $a\in A_+$.
\end{proof}
\end{lemma}
For the complementary case of embeddings into infinite-dimensional classical models, similar results are obtained in the upcoming work of \citet{BarnumL20}.

\inlineheading{Examples of embedding into quantum theory.}
Of particular interest are the embeddings into quantum models, since these seem to be what nature provides us with. 
For example, the quantum-error correcting Shor code~\cite{Shor95} maps a single logical qubit onto nine physical qubits ($\psi: \Csa{2} \to \bigotimes_{i=1}^9 \Csa{2}$) in such a way as to allow for a random bit flip ($\sigma_x$) and/or phase flip ($\sigma_z$) on any of the nine physical qubits without affecting the encoded logical information.
Here both $\psi$ and $\varphi$ take the form $X\mapsto V X V^\dagger$, where $V$ is an isometry.

A second example is the inclusion of $n$-level classical probability theory within $n$-level quantum theory. There is a positive unital linear map $\psi: \reals^n \to \Csa{n}$ -- specifically onto the $n\times n$ diagonal matrices $\psi: (p_1,\ldots,p_n)\trans\mapsto \rho := \sum_i p_i \ketbra{i}{i}$ for some choice of basis $\{\ket{i}\}_{i=1\ldots n}$. 
Since diagonal elements are never negative for (non-negative multiples of) valid quantum states, $\varphi\conj: \rho \mapsto \left(\langle 0|\rho|0\rangle, \ldots \langle n|\rho|n\rangle  \right)\trans$ is also a positive map.

The models that can be embedded into quantum theory are not limited to classical theory, and quantum theory of a lower dimension.
For example, the $d$-dimensional {\em spin-factor} models $\mathcal{B}_d$, whose normalized states are given by $d$-dimensional balls $B_d := \{ (1, \vec{x})\,\, |\,\, \lVert \vec{x} \rVert_2 \leq 1 \}\subseteq \reals\oplus\reals^d$ often arise as foil theories to quantum theory, generalizing the $3$-dimensional real ``Bloch ball'' representation of a qubit into higher dimensions (see e.g.,\ \cite{PawlowskiW12,PaterekDB10,GarnerMD17,KrummM19}).
The effect space  of $\mathcal{B}_d$ is 
 $\left(\reals\oplus\reals^d, B_{d+}, \vec{u}_d \right)$ 
 where $B_{d+}: = \{\left(n, \vec{x}\right)\trans \in \reals\oplus\reals^d \;|\; n\geq 0,\, \lVert \vec{x} \rVert_2 \leq n \}$ 
 and $\vec{u}_d=\left(1,0,\ldots 0 \right)\trans$.
However, a spin factor {\em can} be embedded into complex quantum theory: specifically, into $\mathcal{Q}_{2^{d/2}}$ for even $d$ and $\mathcal{Q}_{2^{\left(d-1\right)/2}}$ for odd $d$ (see \cite{Tsirelson87,KleinmannOSW13,BarnumGW16}).

What about {\em gbits}~\cite{Barrett07}, i.e.,\ models $\mathcal{A}$ with square state spaces that arise in quantum information theory as marginals of hypothetical maximally nonlocal Popescu-Rohrlich boxes~\cite{PopescuR94}? 
The {\em Holevo projection~\cite{Holevo82}} achieves some sort of simulation of those models on \emph{classical} four-level models $\mathcal{B}= \mathcal{C}_4$. 
Namely, the gbit effects $A_+$ are embedded via some map $\varphi$, such that the image $\varphi(A_+)$ consists of the effects $(x_1,x_2, x_3,x_4)\in B_+$ where $x_1+x_2=x_3+x_4$. 
Hence, there is a dual map $\varphi^*$ that maps the classical states (elements of the tetrahedron of four-outcome probability vectors) onto gbit states. 
Could there be some corresponding $\psi$ that maps gbit states to classical states such that all probabilities are reproduced? 
Since the four corner states of the gbit are pairwise perfectly distinguishable, this is only possible if the four deterministic distributions are contained in the image $\psi(\Omega_A)$. 
But the image of a two-dimensional square under any \emph{linear} map $\psi$ cannot contain three linearly independent elements, 
 hence $\psi$ must be non-linear. 
In other words, the Holevo construction violates {\em procedural independence}: almost every gbit state has \emph{infinitely many} classical states that simulate it,
 and which one is the case depends on the specific preparation procedure.

\inlineheading{Every embedding into quantum theory.}
Is there a structural reason that the attempt to embed the gbit failed, but classical theory and the spin factors succeeded? 
To answer this, we first introduce a formal way of saying that an embedding into quantum theory should not be ``unnecessarily large'':

\begin{definition}[Minimal embedding]
\label{def:Min}
An embedding of a model $\mathcal{A}$ into $n$-dimensional quantum theory $\mathcal{Q}_n$ is \textbf{minimal} if there does not exist any $m<n$ such that $\mathcal{A}$ can be embedded into $\mathcal{Q}_m$.
\end{definition}
When we embed into quantum models, we may always choose the smallest possible Hilbert space dimension. Thus, we will henceforth restrict our attention to minimal embeddings. This has the following consequence.
\begin{lemma}
\label{lem:QMinimal}
If an embedding of a model $\mathcal{A}$ into $\mathcal{Q}_n$ is minimal, then there exists some state $\omega\in\Omega_A$ such that the quantum state $\psi\!\left(\omega\right)$ has full rank.
\end{lemma}
\begin{proof}
Let $\omega\in A_+$ such that $m:={\rm rank}(\psi(\omega))$ is maximal, and suppose that $m<n$. Let $S:={\rm supp}(\psi(\omega))$ (an $m$-dimensional subspace of $\comp^n$), and suppose there is some $\rho\in A_+$ with ${\rm supp}(\psi(\rho))\not\subseteq S$. Since ${\rm supp}(\psi(\frac 1 2 \omega+\frac 1 2 \rho))\supseteq {\rm supp}(\psi(\omega))+{\rm supp}(\psi(\rho))\supsetneq S$, this implies ${\rm rank}(\psi(\frac 1 2 \omega+\frac 1 2 \rho))> m$ which is a contradiction. Thus, ${\rm supp}(\psi(\rho))\subseteq S$ for all $\rho\in A_+$, and we can restrict the embedding to $S$ (and thus to $\mathcal{Q}_m$) in an obvious way.
\end{proof}	

We can now prove the main result of this letter: the probabilistic models that can be embedded in quantum theory correspond to the Euclidean special Jordan algebras~\cite{JordanvNW34,AlfsenS03,McCrimmon04,Wilce09,Baez12,Wilce12,BarnumWilce14,BarnumH19}.
We start with two technical lemmas, the first proven in similar form in~\cite{EffrosS79,Idel13}:
\begin{lemma}
\label{lem:EffSt}
For every minimal embedding of a model $\mathcal{A}$ into finite-dimensional quantum theory $\mathcal{Q}_n$, the corresponding projector $P=\varphi\psi\conj$ satisfies
\begin{align}
   P\!\left(x\sjord y\right)=x\sjord P\left(y\right) \mbox{ for all }x\in\varphi(A),\, y\in B.
\end{align}
Hence, $\varphi(A)\equiv P(B)$ is closed under the Jordan product $\sjord$, and $\left(P\!\left(B\right),\, \sjord\right)$ is a special Euclidean Jordan algebra.
\end{lemma}
\begin{proof}
Due to Lemma~\ref{lem:EmbeddingFacts}, $\dim A<\infty$.
First, we show that all $x\in\varphi(A)$ satisfy $P(x^2)=x^2$. 
From Lemma~\ref{lem:QMinimal}, there exists some full-rank fixed state $\rho = P^*\!\left(\rho\right)$, and hence $\tr[P(x^2)\rho]=\tr[x^2 P^*(\rho)]=\tr(x^2\rho)$,
 such that $\tr(\Delta\rho)=0$ for $\Delta:=P(x^2)-x^2$.
Due to {\em Kadison's inequality}~\cite{Kadison52}, $P(z^2)\geq P(z)^2$, hence $\Delta\geq 0$.
Thus, $\tr(\Delta\rho)=0$ is only possible if $\Delta=0$ since $\rho$ is positive definite.

Now let $x\in\varphi(A)$, $y\in B$, and $t\in\mathbb{R}$ be arbitrary, and set $z:=tx+y$. We thus have $x=P(x)$ and $x^2=P(x^2)$.
Since $P$ is positive and unital (Lemma~\ref{lem:EmbeddingFacts}), Kadison's inequality gives
$2t P(x\sjord y)+P(y^2)\geq 2 t P(x)\sjord P(y)+P(y)^2$
for all $t\in\mathbb{R}$. But if $v=v^\dagger$ and $w=w^\dagger$ such that $tv+w\geq 0$ for all $t\in\mathbb{R}$, then $v=0$ (to see this, multiply from left and right by eigenvectors of $v$). Thus, we conclude that the terms linear in $t$ must be equal, and so $P(x\sjord y)=P(x)\sjord P(y)$.

Thus, if $x,y\in P(B)$ then $x\sjord y=P(x)\sjord P(y)=P(x\sjord y)\in P(B)$, and hence $\left(P\!\left(B\right), \, \sjord\right)$ is a Jordan subalgebra of $\Csa{n}$, inheriting the properties of being special and Euclidean from $\Csa{n}$.
\end{proof}

\begin{lemma}
\label{lem:ConeOfSquares}
For every minimal embedding of a model $\mathcal{A}$ into finite-dimensional quantum theory, we have
\begin{align}
P\!\left(B_+\right) & = \{x^2\,\,|\,\,x\in P\!\left(B\right)\}.
\end{align}
\end{lemma}
\begin{proof}
The right-hand side equals the cone of squares $\mathcal{J}_+$ of $\left(P\!\left(B\right), \sjord\right)$ due to Lemma~\ref{lem:EffSt}.
To show $\mathcal{J}_+\subseteq P(B_+)$, let $y:=x^2$ with $x\in P(B)$. 
Then $0\leq y=x\sjord P(x)=P(x\sjord x)=P(y)$ (using Lemma~\ref{lem:EffSt}), and thus $y\in P(B_+)$. 
Meanwhile, using the Hilbert-Schmidt inner product $\inn{x}{y} = \tr(xy)$ to identify $B$ with $B\conj$,
 the cone $\mathcal{J}_+$ is self-dual~\cite{FarautK94} (i.e.,\ $\mathcal{J}_+=\mathcal{J}_+\conj$) since this inner product makes $\mathcal{J}$ Euclidean.
Let $y\in P(B_+)$.
Then, for all $x\in P(B)$, $\inn{x^2}{y} = \tr(x^2 y)\geq 0$ since $x^2\geq 0$ and $y\geq 0$, and thus $y\in \mathcal{J}_+^* \equiv \mathcal{J}_+$, and thus $P(B_+)\subseteq \mathcal{J}_+$.
Hence, $P\!\left(B_+\right) = \mathcal{J_+} = \{x^2\,\,|\,\,x\in P\!\left(B\right)\}$.
\end{proof}

Thus, we state the main theorem of this letter:-
\begin{theorem*}
\label{thm:Simulate}
A model can be embedded into finite-dimensional quantum theory if and only if it corresponds to a Euclidean special Jordan algebra.
\begin{proof}
For the {\em only if} direction, we can choose a minimal embedding $\varphi: A\to \Csa{n}$. From Lemma~\ref{lem:EmbeddingFacts} (iv) and Lemma~\ref{lem:ConeOfSquares}, it follows that $\varphi(A_+)=\{x^2\,\,|\,\,x\in\varphi(A)\}$, hence $A$ is order-isomorphic to the probabilistic model of the special Euclidean Jordan algebra $(P(B),\sjord)$.
 To show the {\em if} direction, we use that such algebras can be exhaustively listed~\cite{JordanvNW34}. Appropriate embedding maps exist for all of these spaces~\cite{Tsirelson87,BarnumGW16,KleinmannOSW13}, and their direct sums.
\end{proof}
\end{theorem*}
\inlineheading{Bipartite correlations.}
\Citet{BarnumBBEW10} show that any bipartite system that looks locally like quantum theory can only admit non-signalling correlations that are quantumly realizable. 
\Citet{KleinmannOSW13} extend this to systems whose local models can be embedded into quantum theory as in our \cref{def:EmbeddingMap}.
Our result thus implies that any non-signalling composite $AB$ of Jordan--algebraic models $A$ and $B$ (e.g.,\ of quaternionic quantum theory) can only contain correlations that can be reproduced within standard complex quantum theory (even if $AB$ is not itself embeddable into quantum theory).

\inlineheading{Embedding and decoherence.}
Recall that the projector $P=\varphi\psi\conj$ (resp.\ $P\conj=\psi\varphi\conj$) maps the set of all quantum effects (states) onto an embedding of the Jordan algebra effect space $A$ (state space $A\conj$), i.e.\ $P\!\left(\Csa{n}\right)=\varphi(A)$ and $P\conj(\Csa{n})=\psi(A\conj)$.
In the special case where $A$ is classical (and standardly embedded as in the second embedding example further above), this is exactly a \emph{decoherence process}, where $P=P\conj$ removes the off-diagonal elements (in some given basis).

Are there analogous ``decoherence'' processes for the (minimal) embeddings of other Jordan--algebraic models $\mathcal{A}$? 
If so, these would have to be completely positive~\cite{Stinespring55,Choi75,NielsenC00} (i.e.\ physically realizable) unital maps $Q$ with $Q^2=Q$ (decohering twice is the same as decohering once), with $Q(\Csa{n})=\varphi(A)$. 
Since $Q\conj(\Csa{n})=\psi(A\conj)$ and $(Q\conj)^2=Q\conj$, the map $Q\conj$ has a full-rank fixed point, and since $Q^2=Q$, the set of fixed points of $Q$ equals $Q(M_n(\mathbb{C}))=\varphi(A)+i\varphi(A)$, where $M_n(\mathbb{C})$ denotes the complex $n\times n$ matrices.
But according to~\cite[Thm.\ 6.12]{Wolf12} (see also~\cite{Asias02}), the fixed-point sets of such completely positive maps $Q$ are ${}\conj$-subalgebras of $M_n(\mathbb{C})$, and thus isomorphic to standard complex quantum theory with superselection rules (including classical theory). Thus, the other, more exotic projectors are ruled out:
\begin{corollary}
\label{eq:Decoherence}
The only quantum--embeddable probabilistic models which can result from a physical decoherence map are the classical state spaces, and standard complex quantum theory with superselection rules.	
\end{corollary}

\inlineheading{Classical embeddings.}
A second corollary of the Main Theorem is a straightforward characterization of all models that can be embedded into classical theory.
\begin{corollary}
A model can be embedded into finite-dimensional classical theory if and only if it is classical.
\begin{proof}
Suppose that $\mathcal{A}$ can be embedded into some $\mathcal{C}_n$ via maps $\varphi,\psi$. 
Since $\mathcal{C}_n$ can be embedded into $\mathcal{Q}_n$ via some $\varphi',\psi'$, this gives us an embedding of $\mathcal{A}$ into $\mathcal{Q}_n$ via $\varphi'\circ\varphi,\psi'\circ\psi$. 
Hence, due to the Main Theorem, $\mathcal{A}$ must correspond to a Euclidean special Jordan algebra. 
But Lemma~\ref{lem:EmbeddingFacts}(iii) tells us that $\varphi(A_+)=\varphi(A)\cap C_+$, where $C_+$ is the polyhedral cone~\cite{AliprantisT07} of classical effects. Hence $A_+$ must be a polyhedral cone too, i.e.\ $A_+$ contains only a finite number of extremal effects. But the only Jordan-algebraic effect spaces with finitely many extremal effects are the classical effect spaces.
\end{proof}
\end{corollary}

Embeddings into classical models are of interest in the foundations of quantum mechanics since they formalize a notion of \emph{hidden-variable models}~\cite{HarriganS10}: given some operational theory (e.g.,\ quantum theory), one may ask whether all its statistics can consistently be understood as arising from unknown underlying classical probability  distributions and response functions. 
In this case, our notion of procedural independence is identical to the well-known condition of \emph{non-contextuality}~\cite{Spekkens05}: statistically indistinguishable preparation (resp.\ measurement) procedures ought to be represented by identical classical distributions (resp.\ response functions)~\cite{SchmidSWKW19}. 
Famously, there are no non-contextual hidden-variable models for quantum theory, and the above corollary shows that this conclusion extends to \emph{all} non-classical probabilistic models that satisfy the no-restriction hypothesis if the hidden-variable model is assumed to be discrete. 
This has recently been proven with alternative methods in~\cite{Shahandeh19} and \cite{BarnumL20}; here it follows as a simple corollary.

\inlineheading{Discussion.}
Our result has significant implications for experimental tests of quantum mechanics. 
Suppose that we isolate a degree of freedom in the laboratory and determine its probabilistic model by attempting to implement as many states and effects as we can, as described in~\citet{MazurekPRS17}.
If we obtain a model whose states and effect spaces are not full duals of each other (e.g.,\ as in stabilizer quantum theory), then this means that we may simply not have looked hard enough to discover all possible states and effects. 
On the other hand, if the model is unrestricted, but is not a standard quantum model, we may ask whether we have uncovered genuine new physics or whether the model could be simply simulated by standard quantum physics.
What our Main Theorem then tells us is: \emph{If an unrestricted model is not Jordan-algebraic, then a quantum simulation is implausible for the same reason that contextual hidden-variable models for quantum theory are implausible}. 
Namely, such a simulation would have to represent statistically identical preparation (or measurement) procedures by different quantum states (or effects) which are fine-tuned~\cite{WoodS15} to yield the exact same statistics.

Our result therefore underlines the physical significance of Euclidean special Jordan models: it characterizes them as the unique unrestricted models that can be embedded into standard quantum theory in a way that respects procedural independence.

\vspace*{0.25em}
\inlineheading{Acknowledgments.}
We are grateful for discussions with Howard Barnum, who drew our attention to Ref.~\cite{Idel13}, and thank Alex Wilce for very helpful and detailed remarks on an earlier draft of this paper.
This research was supported in part by Perimeter Institute for Theoretical Physics. 
Research at Perimeter Institute is supported by the Government of Canada through the Department of Innovation, Science and Economic Development Canada and by the Province of Ontario through the Ministry of Research, Innovation and Science.

\bibliography{embedding}

\end{document}